\documentclass[12pt, english]{article} 
\pdfoutput=1
\usepackage{color}
\usepackage[margin=1in]{geometry} 
\usepackage{graphicx} 
\usepackage{amsmath} 
\usepackage{amsfonts} 
\usepackage{amsthm} 
\usepackage{mathrsfs}
\usepackage{mathtools}
\usepackage{enumitem}
\usepackage{babel}
\usepackage{titlesec}
\usepackage[colorlinks]{hyperref}
\titleformat*{\section}{\normalsize\bfseries}
\titleformat*{\subsection}{\small\bfseries}
\titleformat*{\subsubsection}{\normalize\bfseries}
\titleformat*{\paragraph}{\normalize\bfseries}
\titleformat*{\subparagraph}{\normalize\bfseries}

\usepackage{color}

\usepackage{hyperref}
\newtheorem{theorem}{Theorem}[section]

\newtheorem{corollary}[theorem]{Corollary}

\newtheorem{definition}[theorem]{Definition}
  
\newtheorem{remark}[theorem]{Remark}
\newtheorem{notation}[theorem]{Notation}
\newcommand{\de}{\mathrm{d}}
\DeclareMathAlphabet{\mathpzc}{OT1}{pzc}{m}{it}

\providecommand{\keywords}[1]{\textbf{\textit{Keywords: }}#1}

\title{Ollivier Ricci Curvature of Directed Hypergraphs}

\author{Marzieh Eidi\footnotemark[1]\ \and J\"urgen Jost\footnotemark[1]\ \footnotemark[2]}

\begin{document}

\maketitle
\renewcommand{\thefootnote}{\fnsymbol{footnote}}
\footnotetext[1]{Max-Planck Institut for Mathematics in the Sciences, Leipzig, Germany}
\footnotetext[2]{Santa Fe Institute , Santa Fe , New Mexico , USA}

\renewcommand{\thefootnote}{\arabic{footnote}}

\begin{abstract}
We develop a definition of Ricci curvature on directed hypergraphs and explore the consequences of that definition. The definition generalizes Ollivier's definition for graphs. It involves a carefully designed optimal transport problem between sets of vertices.

\end{abstract}

\section{Introduction}
In Riemannian geometry, the  curvature of a space quantifies its non-flatness.  Among the various curvature notions that are of importance in Riemannian geometry, Ricci curvature  quantifies  this deviation by comparing the average distance between two sufficiently close points and the distance between two small balls around them.  Bounds on curvatures can be used to  connect the  geometry of a Riemannian manifold with its topology, or to control stochastic processes on it. More precisely, a  positive lower bound for the Ricci curvature yields the Bonnet-Myers theorem, which  bounds the diameter of the space in terms of such a lower Ricci bound, the Lichnerowicz theorem for the spectral gap of the Laplacian, a control on mixing properties of Brownian motion and the Levy-Gromov theorem for isometric inequalities and concentration of measures. In view of these strong implications, it is desirable to extend this to  metric spaces that are more general than Riemannian manifolds. With this motivation, several generalized curvature notions have been proposed  for non-smooth or discrete structures. In particular,  Yann Ollivier \cite{Ollivier2009} defined a notion of Ricci curvature on metric spaces equipped with a Markov chain, and extended some of the mentioned results for  positively curved manifolds. His definition  compares the Wasserstein distance between  probability measures supported in the neighborhoods of two given points with the distance between these points. The Wasserstein distance between  two probability measures is defined as the minimal cost needed for  transporting one into the other. That is, an optimal transport problem has to be solved. --  On Riemannian manifolds, this recovers the original notion of Ricci curvature (up to some scaling factort), and at the same time, it naturally applies to discrete metric spaces like graphs.  Recently, this curvature has been applied in  network analysis, to  determine spreading or local clustering in  networks modelled by undirected or directed graphs, see for instance \cite{samal2018comparative}.

On the other hand, many real data sets are naturally modelled by structures that are somewhat more general than graphs, because they may contain relations involving more than two elements. For instance, chemical reactions typically involve more than two substances. This leads to hypergraphs. These hypergraphs may be undirected, as for instance for coauthorship relations, but they may also be directed. Taking the example of  chemical reactions, they are typically not reversible, but rather transform a set of educts into a set of products.  A definion of the Ollivier Ricci curvature of directed graphs was firstly proposed and investigated in \cite{yamada2016ricci} where out-out directions for assigning measures are used. For that, however,  one needs to assume strong connectivity of the underlying directed graphs in order to find transportation plans with finite cost, but this does not hold in many real directed  networks. Therefore, in this paper, we work with  in-out directions, which does not require such a strong assumption. The resulting theory is rather different from that of \cite{yamada2016ricci}. The first extension of  the notion of Ollivier Ricci curvature to hypergraphs was proposed in \cite{asoodeh2018curvature}, using  a multi-marginal optimal transport problem to define curvature. Because of that, the resulting curvature in the end is an analogue of Riemannian scalar rather than Ricci curvature. Also, it does not directly apply to  directed hypergraphs. In this paper, we therefore propose a notion of directed hypergraph curvature  that extends  Ricci rather than scalar curvature. Since in our setting, hyperedges are directed and each direction separates the vertices of the  hyperedge into two classes, similar to directed graphs, we consider a  double marginal optimal transport problem. We study some implications of our definition and then take a closer look at hypergraphs of constant Ricci curvature.

\section{Ricci curvature}

Ricci curvature is a fundamental concept  from Riemannian Geometry  (see for instance \cite{Jost17a}) that more recently has been extended to a discrete setting. 
\\
In a Riemannian manifold $M$ of dimension $N$, Ricci curvature can be defined in several equivalent ways. What is relevant for the extension to the discrete setting is that it measures the local amount of non-flatness of the manifold by comparing the distance between  two small  balls with the distance of their centers when these centers are sufficiently close to each other. If $w$ is a  unit tangent vector at a point $x$ in a Riemannian manifold, $\varepsilon,\ \delta >0$ smaller than the injectivity radius of $M$ and $y$ is the endpoint of $\exp_x\delta w$ and hence has distance $\delta$ from $x$ and  $S_x$ is the sphere of radius $\varepsilon$ in the tangent space at $x$ (and hence $\exp_xS_x$ is the sphere of radius $\varepsilon$ around $x$ in the manifold, then if $S_x$ is mapped to $S_y$ using parallel transport, the average distance between a point of $\exp_xS_x$ and its image in $\exp_yS_y$ is
\[\delta \left(1-\frac{\varepsilon^2}{2N}\mathrm{Ric}\ (w,w)+O(\varepsilon^3+\varepsilon^2\delta)\right)\]
when $(\varepsilon,\delta)\to 0$. This follows from standard Jacobi field estimates. These estimates involve the sectional curvature, but summing over all directions orthogonal to the geodesic connecting $x$ and $y$ results in a Ricci curvature term. Here, one should think of $\varepsilon$ as being smaller than $\delta$, and $O(\varepsilon^3)$  then simply indicates a higher term, whereas $O(\varepsilon^2\delta)$ is needed when the Ricci curvature is not constant.
If balls are used instead of spheres, the scaling factor is $\frac{\varepsilon^2}{2(N+2)}$ instead of $\frac{\varepsilon^2}{2N}$\cite{ollivier2011visual}.
\begin{center}
\hspace{-0.2 cm}\vspace{-0.01 cm} \includegraphics[width=4 cm]{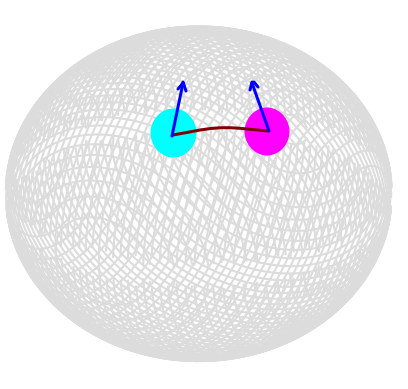} 
\hspace{0.4 cm}\vspace{-0.01 cm} \includegraphics[width=5.7 cm]{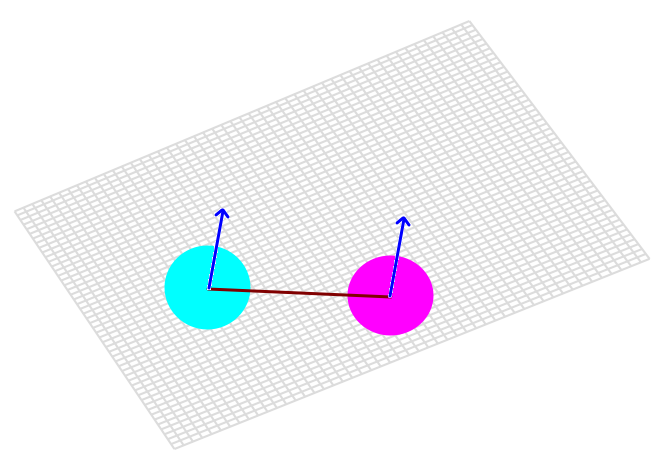} 
\hspace{0.7 cm}\vspace{-0.01 cm} \includegraphics[width=5.5 cm]{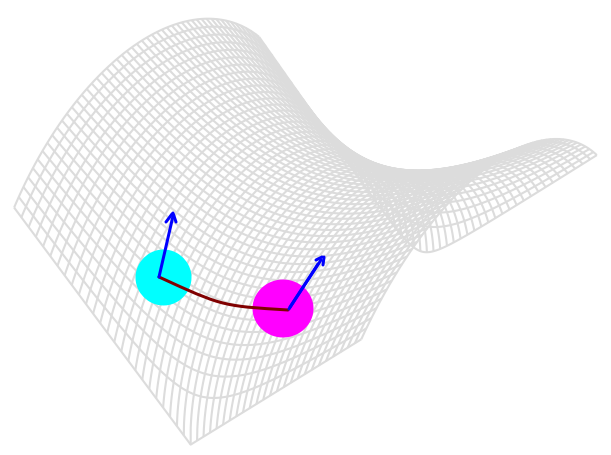} 
\end{center}

Thus, if balls in average are closer than their centers (left figure) , Ricci curvature in the direction of $xy$ is positive. If the manifold is locally flat, Euclidian (middle figure), then the two distances coincide. Most manifolds, however, are locally negatively curved (right figure). \cite{Lohkamp92}.

This local characterization is the key property for defining Ricci curvature notions in more general settings than smooth manifolds. In 2007, Ollivier defined a notion of Ricci curvature, called Ollivier (coarse) Ricci curvature,  on  metric spaces equipped with a random walk $m$:

\begin{definition}\cite{Ollivier2009}
Let $(X,d)$ be a metric space with a random walk $m$, let $x,y\in X$ be two distinct points. The Ricci curvature of $(X,d,m)$ in the direction $(x,y)$ is
\[\kappa(x,y):=1-\frac{W_1(m_x,m_y)}{d(x,y)}\]
where $W_1$ is the 1-Wasserstein distance between $m_x$ and $m_y$ on $X$:
\[W_1(m_x,m_y):=\inf_{\mathcal{E}\in \Pi(m_x,m_y)}\int_{(x,y)\in X\times X} d(x,y)\de \mathcal{E}(x,y)\]
and $\Pi(m_x,m_y)$ is the set of measures on $X\times X$ (coupling between random walks projecting to $m_x$ and $m_y$).
\end{definition}
Recall that if $(X,d)$ is a Polish metric space equipped with its Borel $\sigma$-algebra, a random walk $m$ on $X$ is a family of probability measures $\{m_x|x\in X\}$ satisfying the following conditions :
\begin{itemize}
\item The map $x\to m_x$ is measurable.
\item Each $m_x$ has finite first moment.
\end{itemize}
So, here instead of taking metric balls around two close enough points we consider the well-known Wasserstein distance, transportation or Earthmover distance, between two probability measures $m_x$ and $m_y$  corresponding to two random walks starting at $x$ and $y$ respectively. When $(X,d,m)$ is Riemannian manifold equipped with its Riemannian volume measure, this notion recovers the Riemannian Ricci curvature in the direction of $xy$ (up to some scaling factor)\cite{Ollivier2009}.

In network analysis, this measure is  a very useful tool to determine clustering and coherence in the network \cite{jost2014ollivier,samal2018comparative}, and since it is based on Markov chains, it is very well suited for capturing diffusion and stochastic process in the network.

In order to extend this notion to hypergraphs, an appropriate  definition of random walks on  hypergraphs is needed  which corresponds to the multi marginal transport problem. 
	Recall that an undirected hypergraph $H=(V,E)$ consists of a set $V$ of vertices and a multiset $E$ of subsets of $V$, called hyperedges $(\forall e\in E, |e|\leq |V|)$. Therefore, as a generalization of edges in graphs which connect two vertices, hyperedges represent connections between any number of vertices. In \cite{asoodeh2018curvature}, this lead to 
\begin{definition}
[Coarse scalar curvature]
For a collection of $n$ points $X_n:=\{x_1,\ldots,x_n\}$ in a metric space $(X,d)$ with random walk $m:=\{m_x|x\in X\}$, coarse scalar curvature of $X_n$ is defined as : 
\[\kappa(X_n):=1-\frac{W_1(X_n)}{c(x_1,x_2,\ldots,x_n)}\]
where $W_1(X_n)$ is the minimum of the multi-marginal optimal transport problem and  is equivalent  to the minimum of the Wasserstein barycenter problem:

\[W_1(X_n)=\inf_{\nu \in \Pi(m_x,\nu)}\sum W_1(m_x,\nu)\] 
and $c(x_1,\ldots,x_n)=\inf_{z\in X}\sum_{i=1}^n d_X(x_i,z)$\\

If $\{x_1,\ldots,x_n\}$ are the vertices of a hypergraph connected by a hyperedge $e$, then the above formula for scalar curvature is  an extension for the edge ricci curvature in an undirected graph\cite{jost2014ollivier} (when $n=2$ in every hyperedge). 
\end{definition}
\section{Transport plans and curvature of directed hypergraphs}
 Similarly as in directed graphs, in directed hypergraphs, every hyperedge $e$ represents a directional relation between two subsets $A$ (tail), $B$ (head) of vertices. When  $\forall e\in E, |A|=|B|=1$, then a directed hypergraph is simply a directed graph.\\
We shall now present the basic definition on which this paper rests. 
\begin{notation}
For $x_i\in A$,  $d_{x_i^{in}}$ is the number of incoming hyperedges  to $x_i$ (those hyperedges in which include $x_i$ in the head set of their vertices),  and for $y_j\in B$,  $d_{y_j^{out}}$ is the number of outgoing hyperedges from $y_j$ (those hyperedges which have $y_j$ in the tail set of their vertices).
\end{notation}

\begin{definition}
	Let $H=(V,E)$ be an unweighted directed hypergraph and $e\in E$ an arbitrary directed hyperedge such that $A=\{x_1,\ldots, x_n\}\xrightarrow{e}B=\{y_1,\ldots, y_m\}$ $(n,m\leq |V|)$. We define the Ollivier Ricci curvature of this hyperedge as
	\[\kappa(e):=1-W(\mu_{A^{in}},\mu_{B^{out}})\]
	where the  probability measures $\mu_{A^{in}}$ and $\mu_{B^{out}}$ are  defined on $V$ as follows:
	
	$\mu_{A^{in}}=\sum_{i=1}^n \mu_{x_i}$ where $ \forall 1\leq i\leq n $ and $\forall z\in V(H)$
	\[\mu_{x_i}(z)=\begin{dcases}
	0& z=x_i\quad \& \quad  d_{x_i^{in}}\neq 0\\\frac{1}{n}& z=x_i\quad\& \quad d_{x_i^{in}}=0\\\sum_{e^\prime:z\to x_i}\frac{1}{n\times d_{x_i^{in}}\times \#tail(e^\prime)}& z\neq x_i\quad \&\quad  \exists e^\prime:z\to x_i\\0 &z\neq x_i\quad \& \quad \not\exists e^\prime :z\to x_i
	\end{dcases}\]
and 	likewise \\
	\\$\mu_{B^{out}}=\sum_{j=1}^m\mu_{y_j}$ where $\forall 1\leq j\leq m, z\in V(H)$:
	 \[\mu_{y_j}(z)=\begin{dcases}
	 0& z=y_j\quad \& \quad  d_{y_j^{out}}\neq 0\\\frac{1}{m}& z=y_j\quad \& \quad  d_{y_j^{out}}=0\\\sum_{e^\prime:y_j\to z}\frac{1}{m\times d_{y_j^{out}}\times \#head(e^\prime)}& z\neq y_j\quad \& \quad  \exists e^\prime:y_j\to z\\0 &z\neq y_j\quad \& \quad  \not\exists e^\prime :y_j\to z
	 \end{dcases}\]
	 and $W(\mu_{A^{in}},\mu_{B^{out}})$ is the 1-Wasserstein, optimal transportation or earth mover distance  distance between these two discrete measures as follows: 
	 \[W(\mu_{A^{in}},\mu_{B^{B}})=\min \sum_{u\to A}\sum_{B\to v}d(u,v)\mathcal{E}(u,v)\]
	 where $\mathcal{E}(u,v)$ represents the amount of the mass that should be moved from $u\in A^{in}(u\to A)$ to $v\in B^{out}(B\to v)$, $d(u,v)$ is the minimum number of directed hyperedges to be passed for going from $u$ to $v$ and the minimum is taken over all couplings $\mathcal{E}$ between $\mu_{A^{in}}$ and $\mu_{B^{out}}$ which satisfy
	 \[\sum_{u\to A}\mathcal{E}(u,v)=\sum_{j=1}^m \mu_{y_j}(v)\text{ and } \sum_{B\to v}\mathcal{E}(u,v)=\sum_{i=1}^n \mu_{x_i}(u)\]
\end{definition}
For example, in the directed hypergraph depicted below, for computing the curvature of the yellow hyperedge, we assign masses and holes to neighbours of the left set which includes  3 vertices (separated by dots)  and the right set (which includes 2 vertices and separated in the right ) as follows : 
\begin{center}
\vspace{-0.01 cm} \includegraphics[width=5 cm]{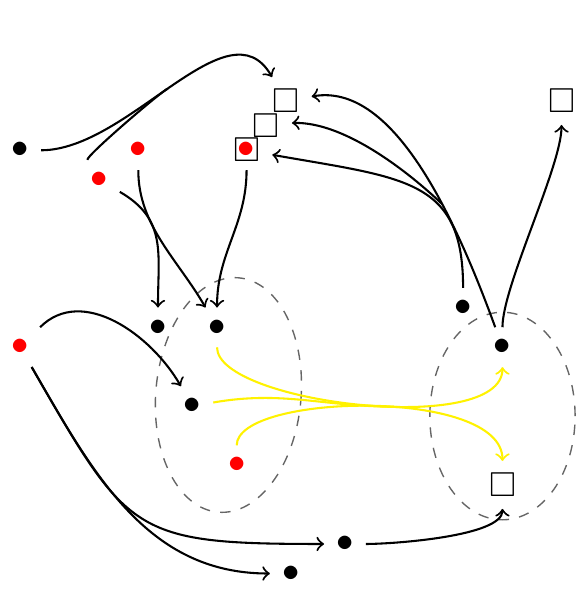}
\end{center}
It ia easy to check that in (any) optimal  optimal transport plan, $1/12 $ of the mass need not to be moved and $1/3$ of it is moved with distance one, $1/6$ is moved with distance two and the remained part is moved with distance three. Hence the  curvature is $-11/12$.
We also point out that while optimal transport plans always exist in our finite setting, they need not be unique. \\
We can construct for each  directed hypergraph a corresponding directed graph. That graph has the same set of vertices as the directed hypergraph and for each hyperedge, we draw an edge from each vertex in its tail to every vertex in its head. Thus, a directed hyperedge  $A=\{x_1,\ldots, x_n\}\xrightarrow{e}B=\{y_1,\ldots, y_m\}$ has a corresponding set of directed edges with $nm$ elements. Note, however, that  there might be directed graphs that correspond to more than one directed hypergraph.
\begin{theorem}\label{lower}
The curvature of a hyperedge $e:A=\{x_1,\ldots,x_n\}\to B=\{y_1,\ldots, y_m\}$ is bounded from below by the minimum of the curvatures of directed edges in its corresponding directed graph.
\end{theorem}
\begin{proof}
Let $\mathcal{E}_{ij}$ be the optimal transport plan for the edge $e_{ij}:x_i\to y_j$, i.e.
\[W(\mu_{x_i^{in}},\mu_{y_j^{out}})=\sum_{u,v\in V}d(u,v)\mathcal{E}_{ij} (u,v)\]
Then
\[\mathcal{E}:=\frac{1}{mn}\sum_{i=1}^n\sum_{j=1}^m\mathcal{E}_{ij}\]
has the marginal distributions $\mu_{A^{in}}$ and $\mu_{B^{out}}$. Therefore
\[W(\mu_{A^{in}},\mu_{B^{out}})\leq \sum_{u,v\in V}d(u,v)\mathcal{E} (u,v)=\frac{1}{mn}\sum_{i=1}^n\sum_{j=1}^m W(\mu_{x_i^{in}},\mu_{y_j^{out}})\leq \max_{\substack{{1\leq i\leq n}\\{1\leq j\leq m}}}W(\mu_{x_i^{in}},\mu_{y_j^{out}})\]
and so $\kappa(e)\geq \min_{\substack{{1\leq i\leq n}\\{1\leq j\leq m}}}\kappa(e_{ij})$.
\end{proof}
\begin{remark}
The maximum of the curvatures of directed edges  corresponding to a directed hyperedge is not necessarily an upper bound for its curvature, as one can see from the following picture where the curvature of the hyperedge with red colour is one and the curvature of all it's four corresponding directed edges is $-1/2$ .

\begin{center}
\vspace{-0.01 cm} \includegraphics[width=3 cm]{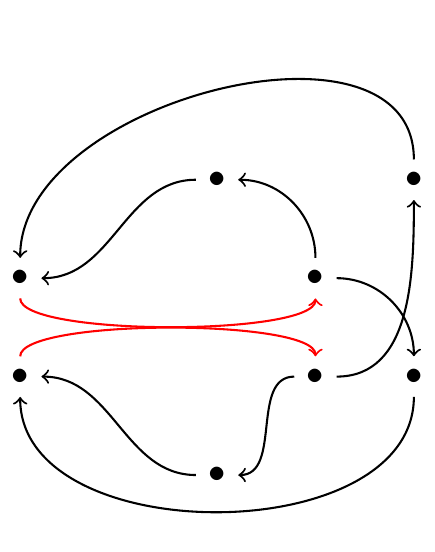}
\end{center}

\end{remark}
\begin{theorem}
For a directed hyperedge $e:A\to B$ we have 
\[W(\mu_{A^{in}},\mu_{B^{out}})\geq \sup \left( \sum_{u\to A} f(u)\mu_{A^{in}}(u)-\sum_{B\to v}f(v)\mu_{B^{out}}(v)\right)\]
where the supremum is taken over all functions on $V(H)$ with $f(u)-f(v)\leq d(u,v)$.
\end{theorem}
\begin{proof}
The proof is similar to the proof of proposition 2.10 in \cite{yamada2016ricci}, which also works for  directed hypergraphs and  some other directions (in-out) for defining measures. Therefore we have :
\begin{align*}
\sum_{u\to A}\sum_{B\to v} d(u,v)\mathcal{E}(u,v)&\geq \sum_{u\to A}\sum_{B\to v} (f(u)-f(v))\mathcal{E}(u,v)\\&=\sum_{u\to A} f(u)\sum_{B\to v} \mathcal{E}(u,v)-\sum_{B\to v}f(v)\sum_{u\to A}\mathcal{E}(u,v)\\
\intertext{and thus according to the projections:}
&=\sum_{u\to A} f(u)\mu_{A^{in}}(u)-\sum_{B\to v}f(v)\mu_{B^{out}}(v)
\end{align*}
and since for all Lipschitz functions on hypergraphs this inequality holds and the left hand side is independent of $f$, we obtain 
\[W(\mu_{A^{in}},\mu_{B^{out}})\geq \sup \left( \sum_{u\to A} f(u)\mu_{A^{in}}(u)-\sum_{B\to v}f(v)\mu_{B^{out}}(v)\right).\]
\end{proof}
\begin{remark}
 Since the distance function on the vertices is not necessarily symmetric, even though the set of Lipschitz functions is always non-empty as it contains constant functions, this supremum might not be achieved.
\end{remark}
Since we use incoming hyperedges to $A$ and outgoing hyperedges from $B$, if $u$ and $v$ are respectively in the support of $\mu_{A^{in}}$ and $\mu_{B^{out}}$ then $d(u,v)\leq 3$. So before giving some bounds for the curvature, we propose another formula for the curvature of a hyperedge which is more intuitive and in some cases much easier to work with.
\\
Let  $\mu_i$  be the amount of  mass that is moved with distance $i (i\leq 3)$ in an optimal transport plans. Then 
\begin{equation}
  \label{mu}
\sum_{i=0}^3 \mu_i=1,\\\sum_{i=1}^3 i\mu_i= W .
\end{equation}
If $\kappa=0$ then $W=1$, we thus  have $\mu_0=\mu_2+2\mu_3$, and so we can define the curvature  of a hyperedge by
\begin{equation}\label{short}
\kappa=\mu_0-\mu_2-2\mu_3.
\end{equation}
 As in  the (undirected) graph case, $\mu_0$ represents the amount of  mass which is not moved in an optimal plan, i.e.,  the amount of the stable mass in directed 3-cycles $(u\to x_i\to y_j\to u)$ or directed loops emerging from  any of the $x_i$s. Although $\mu_1$ (the mass moved with distance one, possibly  through directed quadrangles $(u\to x_i\to y_j\to v, u\to v$)) does not appear in the formula for the curvature, it is an intermediate step for computations of  $\mu_2$  and  $\mu_3$ where  $\mu_2$ is the amount of  mass that should be moved with distance $2$ (possibly through directed pentagons including $x_i$ and $y_j$) and $\mu_3$ is the amount of the mass that is moved with distance $3$ in an optimal plan.
\begin{remark}
While finding the general formula for the computation of  $\mu_1$ (and $\mu_2$) may be difficult, any lower bound for $\mu_1$ (after simply knowing the exact amount of   $\mu_0$) gives us an upper bound for $W$ and therefore a lower bound for the curvature. So, again as in  the graph case we can present upper and lower bounds for the curvature.
\end{remark}
\begin{remark}
The $\mu_i$ can differ between different  optimal transport plans, but \eqref{mu} will always hold. 
\end{remark}
\begin{remark}
The above formula for the curvature also works for edges in undirected graphs.
\end{remark}

\section{Bounds for the curvature}
For an upper bound for the curvature  of a hyperedge we need to control  $\mu_0:$ which corresponds to the stable mass at directed $3$ cycles (triangles in the undirected graph case) and those vertices which are in the intersection of $A$ and $B$  .
\begin{definition}
A directed hyperloop is a directed hyperedge $e:A=\{x_1,\ldots,x_n\}\to B=\{y_1,\ldots, y_m\}$ for which  $ A\cap{B}$ is nonempty.
\end{definition}
In the sequel, we shall see that increasing the number of vertices in this intersection will make the curvature more positive and in the special case where $ A\cap{B}=A=B$, we shall have $\kappa(e)=1$.
\begin{theorem}
For a directed hyperedge $e:A=\{x_1,\ldots,x_n\}\to B=\{y_1,\ldots, y_m\}$  we have
\[\sum_{u\in \mathrm{supp}\ \mu_{A^{in}}(u)\cup \mathrm{supp}\ \mu_{B^{out}} }\mu_{A^{in}}(u)\wedge \mu_{B^{out}}(u)\geq \kappa(e).\]
\end{theorem}
\begin{proof}
This theorem is similar to Theorem $7$ in  \cite{jost2014ollivier}. We simply notice that the number of non-zero elements in this summation coincides with the number of vertices $u$ belonging to a directed $3$-cycle $(u\to x_i\to y_j\to u)$ or $ A\cap{B}$  . Else  $\mu_{A^{in}}(u)$ or $\mu_{B^{out}}(u)$ is zero. Therefore such $u$'s do not play a role in the  summation. 
\end{proof}
\begin{theorem}
The curvature of a directed hyperloop $e:A=\{x_1,\ldots,x_n\}\to B=\{x_1,\ldots,x_n\}$ is one.
\end{theorem}
\begin{proof}
Since in this case all the masses are coincided with all the holes with the same size no mass need to be moved ($\mu_0=1$) and therefore the curvature is one.
\end{proof}
\begin{remark}
In undirected graphs, it has been proven that the local clustering coefficient can control the scalar curvature of any vertex which by definition is obtained by averaging over the Ricci curvature of all the edges connecting to that vertex (see Corollary 1\cite{jost2014ollivier}). Here after fixing the direction of every hyperedge, we encounter  4 different types of triangles which share the property of having a directed edge which goes out from $A$ and comes into the set $B$. Therefore in contrast to the undirected graph case, not all types of directed triangles but only the  presence of directed 3-cycles constituted of vertices of $A$ and  $B$ and those vertices $u$ where $ u\to x_i\to y_j\to u , u\in A^{in}(u\to A)$ and $u\in B^{out}(B\to u)$  increases the curvature of the corresponding hyperedge since they directly affect  $\mu_0$. Also those directed triangles which are constituted by  $A$ and  $B$ and $u$ in such a way that $  x_i\to u,  x_i\to y_j\to u , u\in A^{out}(A\to u)$  and  $u\in B^{out}(B\to u)$ and the ones constituted by  $A$ and  $B$  and $u$ such that $ u\to x_i\to y_j, u\to y_j , u\in A^{in}(u\to A) , u\in B^{in}(u\to B)$ might have impact on the amount of $\mu_2$ which means they can make the curvature less negative. The last type of directed triangles which include vertices of $A$ and $B$, outgoing  vertices of $A$ and incoming vertices to the set $B$ do not affect any of the $\mu_i$s and therefore they do not affect the curvature. For instance in the following directed graph, the triangle constituted of red and green edges affects $\mu_0$ regarding the curvature of the green edge. Both triangles including orange-green and blue-green edges have an effect on $\mu_2$ and the  curvature is not affected by  the presence of the triangle of pink-green edges .
\begin{center}
\vspace{-0.01 cm} \includegraphics[width=5 cm]{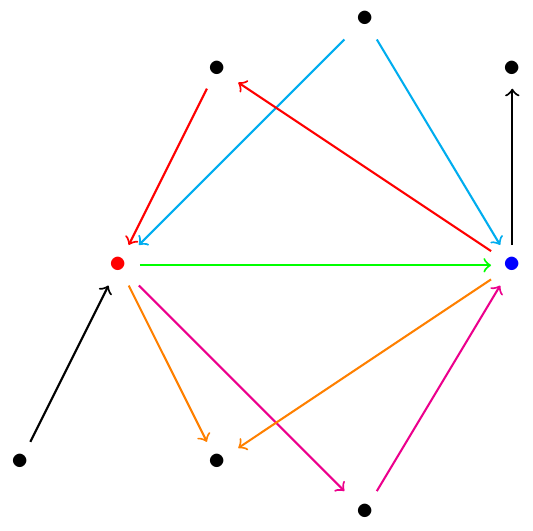}
\end{center}

\end{remark}
\begin{remark}
Directed cycles which have more than 3 connected edges do not affect the curvature of directed hyperedges since they can not make short-cuts for moving  any of the masses to any of the holes 
\end{remark}

 As  already mentioned after computing $\mu_0$, any  non-zero amount for $\mu_1$ would give us an upper bound for $W$. For that,   at least one  incoming neighbour of $A$ should be at  distance one from some outgoing neighbour of $B$. For example, when for some hyperedge, there is at least one hyperedge $e$ from any  $y_j$ to any  $x_i$ $(e:y_j\to x_i)$ or/ and when there is  at least one $x_i$ with $d_{x_i^{in}}=0$ and at least one $y_j$ with $d_{y_j^{out}}=0$, this condition holds and we can present a transfer plan (similar to that in Theorem 3 in \cite{jost2014ollivier}) to obtain a positive lower bound for $\mu_1$. In the same way that trees reach the smallest possible amount of curvature in undirected graphs, here hyperedges in directed hypertrees get the lowest possible number.
 \begin{definition}
 
 A directed loopless hypergraph is a hypertree if 
 \begin{itemize}
 \item[i)] there is  at most one directed path between any two vertices and
 \item[ii)] there is no directed cycle.
\end{itemize}
\end{definition} 
\begin{remark} 
Although these two conditions are equivalent in undirected (hyper)graphs, they do not coincide in the directed case.
 \end{remark}

\begin{theorem}
Let $\{x_1,\ldots,x_n\}\xrightarrow{e} \{y_1,\ldots, y_m\}$ be a hyperedge in a hypertree. 
 
 If $\displaystyle \left\lbrace \begin{array}{l}
\#\{  x_i, 1\leq i\leq n:d_{x_i^{in}}= 0\}=k\\\#\{ y_j, 1\leq j\leq m: d_{y_j^{out}}=0\}=k^\prime\end{array} \right\rbrace \text{ then }{}\kappa(e)=-2+\frac{k}{n}+\frac{k^\prime}{m}.$
\end{theorem}
\begin{proof}
First, since $e$ is in a hypertree, according to the definition $\mu_0=0$. So $\kappa(e)\leq 0$. We shall propose a transfer plan, which gives us an upper bound for $W$, and we shall obtain a lower bound for $W$ based on a single Lipschitz function (defined on the support of $\mu_{A^{in}}$ and $\mu_{B^{out}}$). We shall see that these two bounds  coincide.

 We move $(\frac{k}{n}\wedge\frac{k^\prime}{m})$ of the mass from $k$ $x_i$'s to $k^\prime$ $y_j$'s with distance one. Then if $\frac{k}{n}\geq \frac{k^\prime}{m}$ we move $\frac{k}{n}-\frac{k^\prime}{m}$ of the mass from $x_i$'s with no incoming hyperedges to outgoing neighbours of the $y_j$'s with distance $2$ and if $ \frac{k^\prime}{m}>\frac{k}{n}$ we move $\frac{k^\prime}{m}-\frac{k}{n}$ of the mass at incoming neighbors of $x_i$'s to those $y_j$'s with no outgoing hyperedges with distance two. Then we move the remaining part of the mass with cost $3$. So $W\leq 3-\frac{k}{n}-\frac{k^\prime}{m}$. 

On the other hand, for all $z$ in $V(H)$ we define
\[f(z)=\begin{dcases}
3& \exists 1\leq i\leq n, \exists e:z\to x_i\\
2& \exists 1\leq i\leq n, z=x_i\\
1& \exists 1\leq j\leq m, z=y_j\\
0& \text{otherwise}
\end{dcases}\]
It is easy to check that $f$ is a Lipschitz function on $A\cup B\cup \mathrm{supp} \mu_{A^{in}}\cup\mathrm{supp} \mu_{B^{out}}$, So according to the theorem $3.5$, we have
\begin{align*}
W(\mu_{A^{in}}, \mu_{B^{out}})&\geq \sup \sum_{z\to A}\mu_{A^{in}}(z)f(z)-\sum_{B\to z^\prime}\mu_{B^{out}}(z^\prime)f(z^\prime)\\&\geq 3\left( 1- \frac{k}{n}\right)+2\left(\frac{k}{n}\right)-1\times \frac{k^\prime}{m}-0\times\left(1-\frac{k^\prime}{m}\right)\\&=3-\frac{k}{n}-\frac{k^\prime}{m}.
\end{align*}
So $\displaystyle\kappa(e)=-2+\frac{k}{n}+\frac{k^\prime}{m}$.
\end{proof}

\begin{theorem}
If for all $1\leq i\leq n$, $1\leq j\leq m$, $d_{x_i^{in}}\neq 0$ and $d_{y_j^{out}}\neq 0$ and there is a bijective map $g:\mathrm{supp}\ \mu_{A^{in}}\to \mathrm{supp}\ \mu_{B^{out}}$ such that $g(z)=z^\prime $ and $d(z,z^\prime)=1$ and $\mu_{A^{in}}(z)=\mu_{B^{out}}(z^\prime)$ then $\kappa(e)=0$.
\end{theorem}
\begin{proof}
By assumption, $\mu_0=0$ and the whole mass in any transport plan has to  be moved with distance at least one. Also we know that the bigger $\mu_1$, the lower the cost of the transport between $\mu_{A^{in}}$ and $\mu_{B^{out}}$. But according to the assumption there is a direct (length $1$) path between every pair $(z,z^\prime)$ and it's corresponding hole at $z^\prime$ can be filled with mass at $z$ (no further mass remains at $z$). So $\mu_1=1$ and therefore $W=1$ and $\kappa=0$
\end{proof}
\vspace{1 cm}
\begin{remark}
Neither of the two assumptions $\forall 1\leq i\leq n \ d_{x_i^{in}}\neq 0$ and $\forall 1\leq j\leq m \ d_{y_j^{out}}\neq 0$ and bijectivity of $g$ is  necessary  to have $\kappa=0$. In such cases a more subtle transfer plan is needed. For instance in the following hypergraphs, the red hyperedge has curvature zero.
\begin{center}
\vspace{-0.01 cm} \includegraphics[width=4 cm]{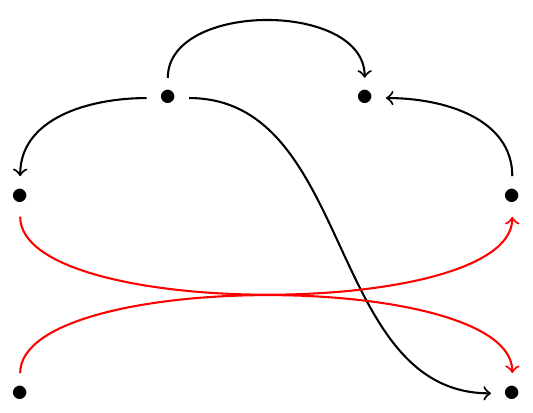}
\end{center}
\end{remark}
\begin{remark}
With the same assumptions as in the previous theorem, but changing the assumption $d(z,z^\prime)=1$  to $d(z,z^\prime)=2$ and assuming there is no directed quadrangle(as before), the curvature will become negative ($-1$). Since in this case the distance between any mass and it's corresponding  hole is $2$ and any hole can not be filled with a mass in which is in lower distance from it.
\end{remark}

\section{Extension and reduction of hyperedges}
\subsection{Removing vertices (Reduction)}

We now want to investigate what happens to the curvature of an edge $e:A=\{x_1,\ldots,x_n\}\to B=\{y_1,\ldots, y_m\}$ if we remove a number $(l,l^\prime)$  of vertices from $A$ $(l\leq n)$ and/or from $B$ $(l^\prime\leq m)$.
Although  curvature  depends on the connections between elements of $\mathrm{supp}\  \mu_{A^{in}}$ and $\mathrm{supp}\ \mu_{B^{out}}$ and removing different vertices from $A(B)$ might have different effects on the curvature, since the amount of  the masses (size of holes) which is assigned to any $x_i(y_j)$ is already determined and is equal to $1/n(1/m)$,  we can give a bound for such changes.

\begin{theorem}
Let $e:A=\{x_1,\ldots,x_n\}\to B=\{y_1,\ldots, y_m\}$. By removing a vertex $x_i$ from $A$ we get $e^\prime: A-\{x_i\}\to B=\{y_1,\ldots, y_m\}$ and we have
\[ |\kappa(e^\prime)-\kappa(e)|\leq \frac{3}{n}.\]
Similarly, by removing $l$ vertices from $A$ $(l<n)$ we have
\[|\kappa(e^\prime)-\kappa(e)|\leq \frac{3l}{n}\]
\end{theorem}
\begin{proof}
The two bounds for the curvature of $e^\prime$ arise from two extreme scenarios which might happen by removing vertex $x_i$ (or $l$ vertices) from $A$.
\begin{enumerate}
\item If the whole mass around $x_i$ (or the $x_i$'s) is in directed loops or directed $3$-cycles including $x_i$ and any of $y_j$'s, and after removing it, its corresponding mass has to be moved with distance $3$ in an optimal plan, then 
\[\begin{dcases}
\kappa(e)=\mu_0-\mu_2-2\mu_3\\\kappa(e^\prime)=\left(\mu_0-\frac{1}{n}\right)-\mu_2-2\left(\mu_3+\frac{1}{n}\right).
\end{dcases}\]
\item If the whole mass around $x_i$ (or the $x_i$'s) was transported with distance $3$ in an  optimal plan and after removing it, the corresponding mass is in the place of directed loops or directed $3$-cycles including vertices of $A-\{x_i\}$ and $B$, then 
\[\kappa(e^\prime)=\left(\mu_0+\frac{1}{n}\right)-\mu_2-2\left(\mu_3-\frac{1}{n}\right)=\kappa(e)+\frac{3}{n}.\]
\end{enumerate}
Therefore we have 
\[\kappa(e)+\frac{3}{n}\geq \kappa(e^\prime)\geq \kappa(e)-\frac{3}{n}.\]
The same argument works for removing $l$ vertices from $A$ and the proof is complete.
\end{proof}
\begin{theorem}
Let  $e:A=\{x_1,\ldots,x_n\}\to B=\{y_1,\ldots, y_m\}$. By removing a vertex $y_j$ from $B$ we get $e^\prime: A\to B-\{y_j\}$ and we have
\[|\kappa(e^\prime)-\kappa(e)|\leq \frac{3}{m}.\]
Analogously, by removing $l^\prime$ vertices from $B$ $(l^\prime<m)$ we have
\[|\kappa(e^\prime)-\kappa(e)|\leq \frac{3l^\prime}{m}.\]
\end{theorem}
\begin{proof}
The argument is  similar to the preceding, the only difference being that we want to fill the corresponding holes with  masses which are  at distance zero or $3$ from them.
\end{proof}
\begin{corollary}
By removing $l$ vertices from the set $A$ and $l^\prime$ vertices from $B$ $(e:A\to B)$ the following relation holds between the curvature of the resulting hyperedge $(e^\prime)$ and the old one:
\[|\kappa(e^\prime)-\kappa(e)|\leq 3\left(\frac{l}{n}+ \frac{l^\prime}{m}\right)\wedge 3 \]
\end{corollary}

\subsection{Adding vertices (Extension)}
Here we want to obtain  bounds for the curvature of a hyperedge  obtained by adding some new vertices to the set $A$ and/or to $B$ and possibly adding new connections between them.
\begin{theorem}
Let $e:A=\{x_1,\ldots,x_n\}\to B=\{y_1,\ldots, y_m\}$. By adding $l$ vertices to $A$ and $l^\prime$ vertices to $B$ we get a hyperedge $e^\prime:A^\prime=\{x_1,\ldots,x_{n+l}\}\to B^\prime=\{y_1,\ldots, y_{m+l^\prime}\}$ with
\[|\kappa(e^\prime)-\kappa(e)|\leq 3\left(\frac{l}{l+n}+ \frac{l^\prime}{l^\prime+m}\right)\wedge 3\]
\end{theorem}
\begin{proof}
Here, since to each $x_i$ in $A^\prime$ we assign $\frac{1}{l+n}$   of the total mass $(=1)$ and $\frac{1}{l^\prime+m}$ of the total hole $(=1)$ to each $y_j$ in $B$, by considering the two extreme scenarios as before we have:
\[\mu_0(e)\to \mu_0(e^\prime)\pm\left(\frac{l}{l+n}+ \frac{l^\prime}{l^\prime+m}\right)\wedge 1\]
and therefore:
\[\mu_3(e)\to \mu_3(e^\prime)\mp\left(\frac{l}{l+n}+ \frac{l^\prime}{l^\prime+m}\right)\wedge 1\]
So the proof is complete.
\end{proof}

\begin{remark}
Reduction and extension of a directed  hypergraph can be occurred in the level of directed hyperedges by adding or removing some hyperedges which connect any sets of vertices which are incoming to A and /or outgoing from B. In special case if we connect all the subsets of A to the all the subsets of B we obtain a directed simplicial complex out of  directed hyperedge (e) which has the same curvature as $e$ since in this construction the distance between any of masses and any of holes and their sizes will not change .
\end{remark}

\section{Directed hypergraphs with constant Ricci curvature}
In this section,  we want to construct  examples of  directed hypergraphs in which the curvature of the hyperedges is constant $(\kappa=1,\kappa=0,\kappa=-2)$. In the case of $\kappa=0$ these (hyper)graphs are called Ricci flat. For brevity,  we also call the others  Ricci $1$ and Ricci $-2$ directed hypergraphs.
\begin{itemize}
\item Ricci 1 directed hypergraphs 
\begin{theorem}\label{Ric1-1}
The vertices of a Ricci 1 directed loopless hypergraph which for every hyperedge $e:\{x_1,\ldots,x_n\}\to\{y_1,\ldots, y_m\}$ does not have any hyperedge in the reverse direction $(\not\exists e^\prime:y_j\to x_i)$, can be divided into $3$ subsets $A, B, C$ such that $A\to B\to C\to A$. This  means that some (not necessarily all) vertices in $A$  are connected to vertices in $B$ via a non-empty collection of directed hyperedges and similarly for the other connections as in the picture. 
\begin{center}
\vspace{-0.01 cm} \includegraphics[width=3.5 cm]{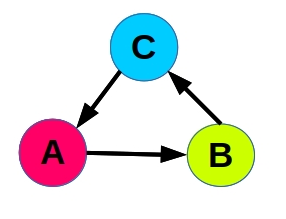}
\end{center}
\end{theorem}
\begin{proof}
Consider a hyperedge $e_1:A_1\to B_1$. Since $\kappa(e_1)=1$
\begin{equation}\label{star}
\mathrm{supp}\ \mu_{A_1^{in}}=\mathrm{supp}\ \mu_{B_1^{out}}=:C_1\quad\text{and}\quad \forall z\in C_1: \mu_{A_1^{in}}(z)=\mu_{B_1^{out}}(z)
\end{equation}
So the diagram related to $e_1$ looks  like 
\begin{center}
\vspace{-0.01 cm} \includegraphics[width=4 cm]{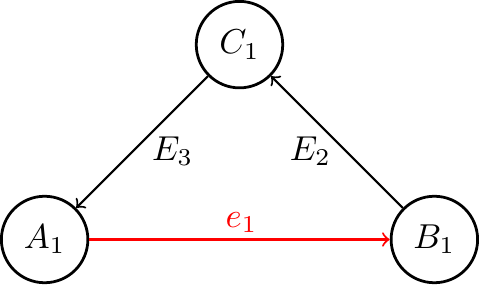}
\end{center}
where  $E_2$ and $E_3$ represent  collections of directed hyperedges. Now, if there is no outgoing hyperedge from $C_1$ other than elements in $E_3$ and there is no incoming hyperedge to $C_1$ other than elements of $E_2$ and there is no outgoing hyperedge other than $e_1$ from $A_1$ and there is no incoming hyperedge to $B_1$ other than $e_1$, then $A_1,B_1$ and $C_1$ would be the desired partitioning set and  since the condition holds for every hyperedge in $E_2$ and $E_3$, the hypergraph is Ricci 1. If any of the above conditions does not hold, we can extend $A_1$ and/or $B_1$ and/or $C_1$ as follows: \\
For instance, let there be  at least one hyperedge going out of $A_1$ other than $e_1$; we call it $e_{OA_1}:A_1\to B_{11}$ and we put $B_2=B_1\cup B_{11}$. Since $\kappa(e_{OA_1})=1$, so $C_2:=\mathrm{supp}\ \mu_{B_{11}^{out}}=\mathrm{supp}\ \mu_{A_1^{in}}\supseteq C_1$. We next consider edges in $E_3$.  If any of them has an endpoint outside $A_1$ and if the set of endpoints of $E_3$ is denoted by $A_2$, then  $A_1\subseteq A_2$. By repeating this process we obtain an increasing sequence of $A_i$'s, $B_j$'s and $C_k$'s. We put $A=\cup A_i$, $B=\cup B_j$ and $C=\cup C_k$.  Obviously, based on the process, elements in $A$ are connected to $B$, $B$ to $C$ and $C$ to $A$ and these $3$ sets are our desired partition.
\begin{center}
\vspace{-0.01 cm} \includegraphics[width=4 cm]{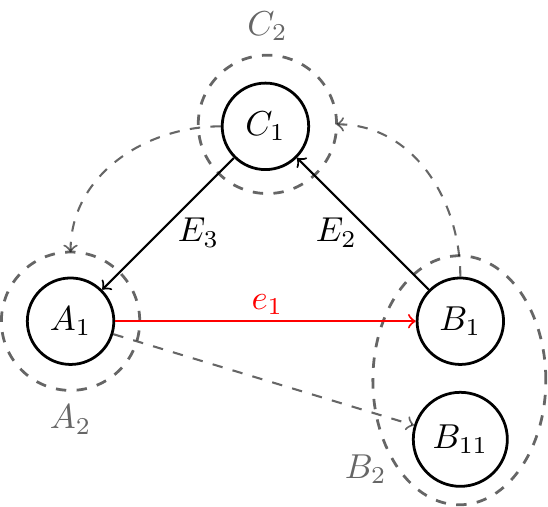}
\end{center} 
 
\end{proof}
\begin{remark}
The converse of this theorem is not necessarily true. For instance,  the following hypergraph  is not Ricci 1 although there is such a partitioning hypergraph:
\begin{center}
\vspace{-0.01 cm} \includegraphics[width=3.5 cm]{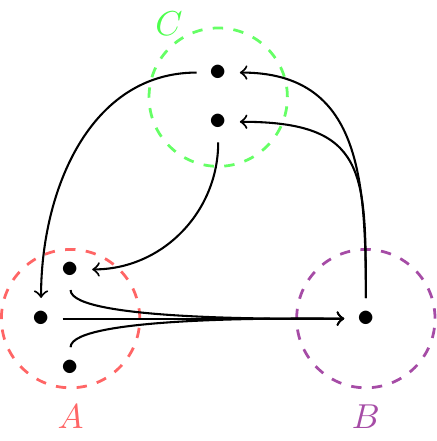}
\end{center}
Instead we have the following:
\end{remark}

\begin{theorem}\label{theorem 2}
If in the corresponding directed graph of a directed (loopless) hypergraph, the set of vertices can be partitioned into $3$ different sets $A,B,C$ such that $A\to B\to C\to A$ and all of the elements in $A$ are connected (via directed edges) to all the elements in $B$ and similarly for the other arrows, then the original directed hypergraph is Ricci 1.
\end{theorem}
Before proving this theorem we state the next theorem.  
\begin{theorem}
A directed (loopless) graph is Ricci 1 iff it's set of vertices can be partitioned into $3$ sets $A,B,C$ such that $A\to B\to C\to A$ and all the vertices in $A$ are connected to all the vertices in $B$ and similarly for the other arrows,  as  in the diagram.
\begin{center}
\vspace{-0.01 cm} \includegraphics[width=3.2 cm]{1c.png}
\end{center}
\end{theorem}
\begin{proof}
$\Rightarrow$The proof is similar to that of Theorem \ref{Ric1-1}. Here, in addition we should have connections between all the vertices of $A$ to all the vertices of $B$ and so on. The reason is that here, for every edge $e:x\to y$, $d_{in}x=d_{out}y$, and the condition that $\mathrm{supp}\ \mu_{x^{in}}=\mathrm{supp}\ \mu_{y^{out}}$ implies that  the tails of incoming edges to $x$   coincide with  the heads of outgoing edges from $y$. So in the resulted partition every vertex in $A$ is connected to every vertex in $B$ and similarly for the connections between other sets the same situation holds. 

$\Leftarrow$ For proving that every edge has curvature 1, we need   that for every edge $e:x\to y$, $d_{in}x=d_{out}y$ and $\mathrm{supp}\ \mu_{x^{in}}=\mathrm{supp}\ \mu_{y^{out}}$ and for every $z$ in this support $ \mu_{x^{in}}(z)=\mu_{y^{out}}(z)$. Since in the partition the  vertices in $A$ are connected to the vertices of $B$ and so on, for every edge the needed conditions  obviously hold. So the directed graph is Ricci 1.
\end{proof}
\begin{proof}[Proof of Theorem \ref{theorem 2}]
Since we have such a partition for the vertices of the corresponding directed graph of this hypergraph, according to the previous theorem, the curvature of all the edges of every directed hyperedge is 1. So their minimum  also has curvature 1. On the other hand, according to Theorem \ref{lower} 
\[\kappa(\text{every hyperedge})\geq \min \kappa(\text{edges in the corresponding directed graph})\]
So for all hyperedges $e$, $\kappa(e)=1$ and the hypergraph is Ricci 1.
\end{proof}
\begin{remark}
It might be possible that the directed hypergraph is Ricci 1,  but as  shown in the following example, its corresponding directed graph is not.
\begin{center}
\vspace{-0.01 cm} \includegraphics[width=2.3 cm]{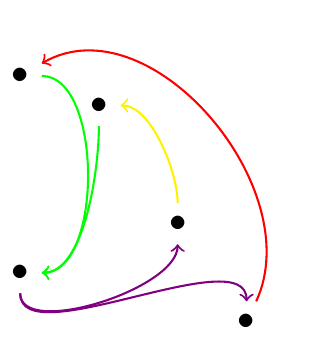}
\end{center}
and the corresponding directed graph is 
\begin{center}
\vspace{-0.01 cm} \includegraphics[width=2.3 cm]{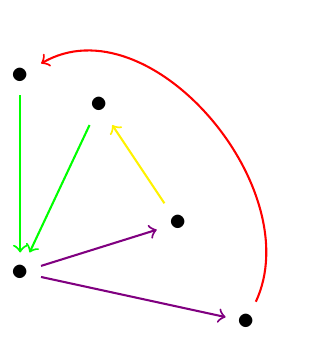}
\end{center}
\end{remark}
\begin{corollary}
If the vertices of a directed hypergraph can be divided into  3 sets, $A$, $B$, $C$ such that all the vertices in these sets are connected to the vertices of the other sets as  shown above, then the hypergraph is Ricci 1. 
\end{corollary}
\begin{proof}
Since we are assuming that all the vertices of $A$ are connected to all the vertices of $B$ and similarly for the other two arrows the same happens, considering every hyperedge  $e:A\to B$ , for any incoming neighbour of $A$ there is a coinciding outgoing neighbour of $B$ and the size of each mass and each hole is the same. Therefore $\mu_0=1$  and the hypergraph is Ricci 1. 
\end{proof}
\begin{corollary}
If in a directed (loopless) hypergraph, the set of vertices can be partitioned into $3$ different sets $A,B,C$ such that $A\to B\to C\to A$ and all of the elements in $A$ are connected (via directed hyperedges) to all the elements in $B$ and similarly for the other arrows, the directed hypergraph is Ricci 1. Here any connection inside any of these sets might violate the constant curvature $1$ for some hyperedges. For instance, in the following directed graph, the curvature of all the edges is 1 but the red edge has curvature $-1$ . 
\begin{center}
\vspace{-0.01 cm} \includegraphics[width=4 cm]{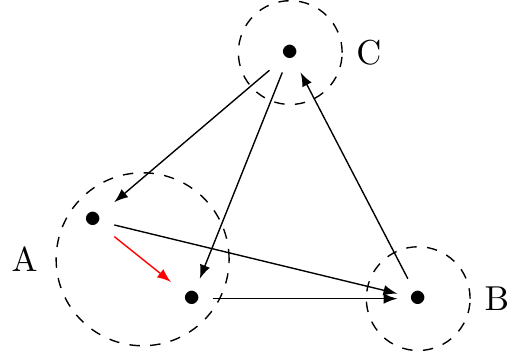}
\end{center}
 
\end{corollary}

\item Ricci flat directed hypergraphs
\begin{theorem}
If the vertices of a directed hypergraph can be divided into two sets $A$ (source) and $B$ (sink)  such that all the vertices in $A$  have outgoing hyperedges and no incoming hyperedges and all the vertices of $B$  have incoming hyperedges and no outgoing hyperedges, then the hypergraph is Ricci flat . 
\begin{center}
\vspace{-0.01 cm} \includegraphics[width=3 cm]{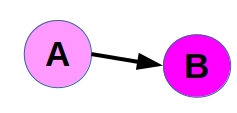}
\end{center}
\end{theorem}
\begin{proof}
Based on the construction, for every hyperedge, the masses are in the  source set $(A)$ which is at distance one from the holes which are in the sink set $(B)$. So $\mu_1$ is equal to $1$ and the hypergraph is Ricci flat . 

\end{proof}
\begin{theorem}
If in a Ricci flat directed hypergraph for every directed hyperedge $e$ there is no incoming hyperedge to its tail set and there in no outgoing hyperedge from its head set, then the set of vertices in this directed hypergraph can be partitioned into two sets $ A$ and $B$ as above.
\end{theorem}
\begin{proof}
We put all the tail sets of all of directed hyperedges in set $A$ and all the head sets of all of directed hyperedges in set $B$. Obviously this is a partitioning of the whole vertices into two sets in which the vertices in $A$ (not necessarily all) are connected to the vertices in $B$ (not necessaries all) and $A$ and $B$ are respectively source and sink sets.
\end{proof}
\begin{theorem}
If the set of vertices of a directed hypergraph can be divided  into 3 sets $A$, $B$,  $C$ such that all the vertices in these sets are connected to the vertices of the other sets as indicated in the diagram, then the hypergraph is Ricci flat. Here, in contrast  to the previous case, the sets are partitioned  into source, saddle and sink sets. The  vertices in a saddle  have both incoming  and outgoing hyperedges. Similar to the previous and the Ricci 1 case, connections inside any of these 3 sets might violate constant curvature along different hyperedges. However, if in addition, the same kind of partitioning with full connections can be found for the vertices of those sets which have inside hyperedges, then presence of these internal connections do not violate flatness .
\begin{center}
\vspace{-0.01 cm} \includegraphics[width=3 cm]{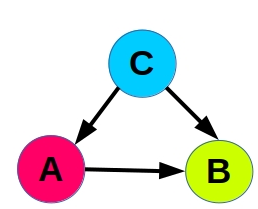}
\end{center}
\end{theorem}
\begin{proof}
Since we are assuming that all the vertices of $A$ are connected to all the vertices of $B$ (and for the 3 partitioning sets the same condition holds), for every hyperedge  $e:A\to B$, the distance between any incoming neighbour of $A$ to any outgoing neighbour of $B$ is one. So $\mu_1$ is equal to $1$ (and obviously by construction $\mu_0=0$). Therefore the hypergraph is Ricci flat . 

\end{proof}
\begin{remark}
If not all the possible connections between these sets exist, even with having this partitioning, Ricci flatness might be violated.
\end{remark}
\begin{remark}
Examples of Ricci flat directed hypergraphs can be constructed in which the set of their vertices is partitioned into 3 sets, but not all the above connections are present. In these hypergraphs, as before the presence of internal hyperedges (Likewise, connections inside each of these sets) might violate the flatnesses can be seen in the next figure.
\begin{center}
\vspace{-0.01 cm} \includegraphics[width=3.5 cm]{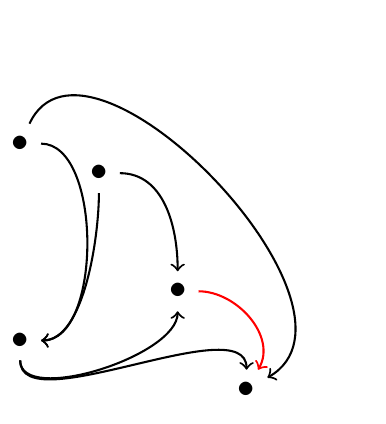}
\end{center}
\end{remark}
\item Ricci negative ($-2$) directed hypergraphs 
\begin{theorem}
If  the set of vertices of a directed hypergraph can be divided into 4 sets, $A$ ,$B$ , $C$  and $D$ such that all the vertices in these sets are connected to the vertices of the other sets as indicated, then the hypergraph is Ricci $-2$ . The presence of internal hyperedges (connections inside each of these sets ) might violate constant curvature along different hyperedges. 
\begin{center}
\vspace{-0.01 cm} \includegraphics[width=3 cm]{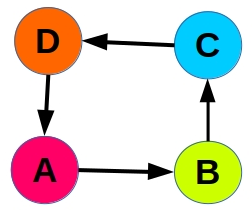}
\end{center}
\end{theorem}
\begin{proof}
It is easy to see that  for every hyperedge  $e:A\to B$, the distance between any incoming neighbour of $A$ to any outgoing neighbour of $B$ is 3. So  all $\mu_i =0$,  except $\mu_3=1$, and so, every hyperedge has curvature $-2$. 
\end{proof}
\begin{remark}
Many Ricci $-2$ directed hypergraphs can be constructed in which the set of their vertices is partitioned into four sets, but not all the above connections exist. 
\end{remark}
\begin{remark}
Although  we have presented some general examples of directed Ricci flat and $-2$ hypergraphs, we still cannot classify them. Also  playing  with \eqref{short} and considering different values of the $\mu_i$s, we can obtain non-negative, negative and  non-positive  curvatures for hyperedges, and possibly some Ricci constant hypergraphs. 
\end{remark}
\end{itemize}

\section{Weighted directed hypergraphs}
We can extend our constructions to weighted directed hypergraphs where the  vertices and hyperedges may both carry weights. The vertices may carry different weights depending on the hyperedges  they are involved in (This can be represented by a vector of the dimension of the hyperedge set with non-negative components. Here, zero means the corresponding hyperedge does not involve that vertex).  For a specified hyperedge whose curvature we want to measure, the weights of its vertices need to be fixed, of course. In this case we denote the vertex and hyperedge weights by $w_v$ and $w_e$ respectively.
\begin{definition}
Let $H=(V,E)$ be  a weighed directed hypergraph and $e\in E$ an arbitrary directed hyperedge such that $A=\{x_1,\ldots, x_n\}\xrightarrow{e}B=\{y_1,\ldots, y_m\}$ $(n,m\leq |V|)$. We define the Ollivier Ricci curvature of this hyperedge as
	\[\kappa(e):=1-W(\mu_{A^{in}},\mu_{B^{out}})\]
	where the probability measures $\mu_{A^{in}}$ and $\mu_{B^{out}}$ are   defined on $V$ as follows:
\[\mu_{A^{in}}=\sum_{i=1}^n\mu_{x_i^{in}}\quad \forall 1\leq i\leq n \quad \text{and } \forall z\in V(H)\]
\[\mu_{x_i^{in}}(z)=\begin{dcases}
0& z=x_i \quad \text{and} \quad d_{x_i^{in}}\neq 0\\
\frac{w_{x_i}}{\sum_{i=1}^n w_{x_i}}& z=x_i \quad\text{and} \quad d_{x_i^{in}}=0\\
\sum_{e^\prime:z\to x_i}\frac{w_{x_i}}{\sum_{x_j\in A} w_{x_j}}\times \frac{w_{e^\prime:z\to x_i}}{\sum_{e}w_{e:z\to A}}\times \frac{w_z}{\sum_{z \text{ tail } e^\prime}w_z}& z\neq x_i\quad\text{and} \quad \exists e^\prime:z\to x_i\\
0&otherwise. 
\end{dcases}\]
Similarly, $\mu_{B^{out}}=\sum_{j=1}^m \mu_{y_j}$, $ \forall 1\leq j\leq m, z\in V(H)$
\[\mu_{y_j^{out}}(z)=\begin{dcases}
0& z=y_j \quad \text{and} \quad d_{y_j^{out}}\neq 0\\
\frac{w_{y_j}}{\sum_{y_i\in B} w_{y_i}}& z=y_j \quad\text{and} \quad d_{y_j^{out}}=0\\
\sum_{e^\prime:y_j\to z}\frac{w_{y_j}}{\sum_{y_i\in B} w_{y_i}}\times \frac{w_{e^\prime:y_j\to z}}{\sum_{e}w_{e:B\to z}}\times \frac{w_z}{\sum_{z \text{ head of } e^{\prime}}w_z}& z\neq y_j\quad\text{and} \quad \exists e^\prime:y_j\to z\\
0&otherwise.
\end{dcases}\]
\end{definition}
\begin{theorem}
Let $\{x_1,\ldots,x_n\}\xrightarrow{e} \{y_1,\ldots, y_m\}$ be a hyperedge in a weighted hypertree. 
 If $\displaystyle \left\lbrace \begin{array}{l}
\#\{  x_i, 1\leq i\leq n:d_{x_i^{in}}= 0\}=k\\\#\{ y_j, 1\leq j\leq m: d_{y_j^{out}}=0\}=k^\prime\end{array} \right\rbrace \text{   then }
{}\kappa(e)=-2+\sum_{i=1}^k\frac{w_{x_i}}{{\sum_{i=1}^n}w_{x_i}}+\sum_{j=1}^{k^\prime}\frac{w_{y_j}}{{\sum_{j=1}^m}w_{y_j}}.$
\end{theorem}
\begin{remark}
The assumption  of theorem 6.4 cannot hold for weighted directed graphs since because of the weights we might have masses which coincide with holes of different sizes. For instance if we consider two directed 3-cycles which have one edge in common, by considering different weights assigned to two other edges in the two cycles, the curvature of the common edge is not one  although we have a 3 set partitioning in which all the connections exists.
\end{remark}

\section{Differences between directed and undirected (hyper)graphs}

\begin{itemize}
\item In directed (hyper)graphs, lower curvature bounds no longer control random walks.
\item Since the Wasserstein distance no longer needs to satisfy a triangle inequality, we cannot define curvatures for vertex sets that are not connected by a hyperedge.  \\
These problems come essentially from the fact that we consider incoming edges at the tail $A$ and outgoing edges at the head $B$ of a hyperedge. In principle, we could of course also consider in-in or out-out relationships instead, but then, we might not always be able to move our masses, and so, curvatures might then become $-\infty$. This can only be avoided if we assume some strong connectedness condition in the directed case ( see for instance \cite{yamada2016ricci}). Such a condition, however, is typically not satisfied in empirical data sets. 
\item The curvature of a directed (hyper)graph might be rather different from that of the underlying undirected (hyper)graph. For instance, in an undirected graph, every edge in a cycle has curvature zero. But a directed cycle where all the edges have the same direction is negatively curved. 

\end{itemize}
\section*{Acknowledgement}

Marzieh Eidi wishes to thank Sima Mehri, Florentin M\"unch and Paolo Perrone for the enlightening discussions.
She also thanks Wilmer Leal for his comments and kind help for drawing hypergraphs. 

\bibliographystyle{plain}

\bibliography{ms}
\end {document}